\documentclass[10pt]{article}
\usepackage[letterpaper]{geometry}
\usepackage{hicss}
\usepackage{times}
\usepackage[none]{hyphenat}
\usepackage{url}
\usepackage{latexsym}
\usepackage{indentfirst}
\usepackage{graphicx}
\graphicspath{{images/}}
\usepackage[
    style=ieee,
  ]{biblatex}
\usepackage{mathtools}  
\usepackage{amsmath} 
\usepackage{amssymb}  
\usepackage{amsthm}  
\usepackage[bookmarks=false]{hyperref}
\usepackage{multicol} 
\usepackage{url} 
\usepackage{siunitx} 
\usepackage{subcaption} 

\usepackage{algorithm}
\usepackage{algpseudocode}

\newcommand*{\hquad}{\hspace{0.5em}}
\DeclareMathOperator*{\argmin}{arg\,min}
\DeclareSIUnit \var { VAr }
\newtheorem{theorem}{Theorem}
\newtheorem{lemma}{Lemma} 

\renewcommand\cite[1]{\parencite{#1}} 

\title{Safe Control of Grid-Interfacing Inverters with Current Magnitude Limits}

\author{Trager Joswig-Jones \\
 University of Washington \\
 {\underline{joswitra@uw.edu}} \\ \And
 Baosen Zhang \\
 University of Washington \\
 {\underline{zhangbao@uw.edu} } \\ 
}

\date{}

\begin{document}
\maketitle
\begin{abstract} 


Grid-interfacing inverters allow renewable resources to be connected to the electric grid and offer fast and programmable control responses. However, inverters are subject to significant physical constraints. One such constraint is a current magnitude limit required to protect semiconductor devices. While many current limiting methods are available, they can often unpredictably alter the behavior of the inverter control during overcurrent events leading to instability or poor performance.

In this paper, we present a safety filter approach to limit the current magnitude of inverters controlled as voltage sources. The safety filter problem is formulated with a control barrier function constraint that encodes the current magnitude limit. To ensure feasibility of the problem, we prove the existence of a safe linear controller for a specified reference. 
This approach allows for the desired voltage source behavior to be minimally altered while safely limiting the current output.

\end{abstract}

\subsubsection*{Keywords:}

current limits, grid-interfacing inverter, stability, safety-critical control, control barrier function


\section{Introduction}  
Electrical power systems are becoming more reliant on renewable energy resources connected to the grid through power electronic devices~\cite{Kroposki_Johnson_Zhang_Gevorgian_Denholm_Hodge_Hannegan_2017}. 
Historically, power electronic inverters were solely designed to convert power from resources, such as wind, solar, and battery energy storage, to be compatible with the AC power system. Now, as these inverter-based renewable energy resources gain higher levels of participation in the grid it is still an open question as to how inverters should be controlled to best support the stability of an AC power system~\cite{guo2019performance, matevosyan2019grid}.

Many inverter control methods have been proposed in literature with the aim of supporting the stability of the grid. One group of methods that have recently gained traction are called grid-forming controllers, which includes droop~\cite{Chandorkar_Divan_Adapa_1993, schiffer2014conditions}, virtual synchronous machine~\cite{driesen2008virtual}, and virtual oscillator control~\cite{dhople2013virtual}. 
Other approaches include training neural-network based controllers~\cite{cui2022reinforcement}, and shaping the frequency response of the system through inverter control~\cite{jiang2021grid, jiang2020dynamic}.

However, many of these approaches assume the inverter model to be ideal (i.e. the device can achieve arbitrary voltage and current outputs) and do not model physical limitations on the devices that are important in practice. 
In reality, controlling an inverters connected to the grid requires a trade-off between system stability and self-protection~\cite{Hart_Gong_Liu_Chen_Zhang_Wang_2022}. 
One of these limitations needed for self-protection is the amount of current that can pass through the switches of inverter devices before they are damaged. 
A common strategy to adapt these approaches for practical implementation is to augment the inverter controller with a current limiting module. While many current limiting methods for grid-interfacing inverters have been proposed, there are still open issues related to how to limit the current when controlling an inverter as a voltage source (as is the case with grid-forming control)~\cite{Fan_Liu_Zhao_Wu_Wang_2022}.

A current limiter saturates the reference of a closed-loop current controller to avoid exceeding the allowed current magnitude. This approach, although simple, changes the behavior of an inverter, and makes its stability and voltage support capability depend on the exact operating conditions~\cite{Xing_Min_Chen_Mao_2021, joswigjones2024optimal}. To avoid these qualitative behavioral changes, voltage limiters have been proposed to directly limit the voltage difference across the inverter filter, such that the current would not exceed its magnitude limits~\cite{VoltageLimit_Bloemink2012, VoltageLimit_Zhou2021}. Virtual impedance methods can also be used to avoid changes in voltage leading to large changes in current that exceed the current magnitude constraint~\cite{VirtualImpedance_Paquette_2015, AdaptiveVirtualImpedance_Wu_2022}. However, these methods require careful design based on detailed knowledge of parameters or system measurements. In addition, they invariably involve a tradeoff between current limiting, and the performance and stability of the inverter. 


To optimize the tradeoff between satisfying current limits and performance, a number of techniques have been proposed. A common approach is to design a ``good'' controller by looking at a system without consideration for the current constraint, then modify its control action as little as possible when current limits are encountered. 
Some research has been presented for virtual impedance methods indicating that the magnitude of the change of the inverter terminal voltage is minimized when the phase angle of the virtual impedance equals that of the passive impedance composed by filter impedance and grid impedance~\cite{Wu_Wang_Zhao_2024}.
This provides a method for selecting the virtual impedance parameters, but requires knowledge of the passive impedance which can change under different operating scenarios and be challenging to obtain exactly.
Another relevant approach is the safety filter framework presented in \cite{Schneeberger_Mastellone_Dörfler_2024}, which is demonstrated with a model of an inverter connected to an infinite bus. This framework can design safety filters (functions that modify control actions) by solving a sum of squares optimization problem and implement the resulting safety filter as a quadratically constrained quadratic program. However, finding these guarantees is computationally difficult and would have to be computed separately for differing inverter systems.
Next, projected-droop control, a control method that explicitly considers the current limits of an inverter, is proposed in ~\cite{Groß_Dorfler_2019}. This approach is designed specifically for the dynamics of a droop controller and is not generally applicable to other control methods. \cite{Groß_Lyu_2023} presents a similar approach that uses primal-dual gradient dynamics to develop a feedback controller for a generic grid-forming control problem and considers a range of constraints. However, this approach does not guarantee constraint satisfaction for all times as the model does not capture the circuit dynamics of the system.

In this work we demonstrate a control barrier function based safety filter method for current limiting of inverter-based resources controlled as a voltage source. The safety filter optimization problem is formulated with a control barrier function constraint that encodes the current magnitude limit. The feasibility of this optimization problem is guaranteed by the existence of a safe linear feedback controller for a given feasible reference. To guarantee the existance of a safe linear controller we use the dynamic properties of an inverter connected to a voltage source through an \emph{RL} branch. This method aims to minimally alter the nominal control action to try to preserve the performance of a nominal controller, while safely limiting the current output at all times. We provide a simple closed-form solution to the safety filter problem that can be applied to a nominal control action resulting in a safe controller. Simulations are performed comparing an unsafe nominal controller to the safety filter controller and a safe linear feedback controller.


The rest of the paper is structured as follows: In Section \ref{section:model}, the system model is introduced and the safety and stability of this model are discussed. Section \ref{section:control} introduces the current magnitude safety filter formulation, which is then shown to always be feasible using proofs provided in Section \ref{sec:proof}. Numerical simulations are presented in Section \ref{section:results} to demonstrate the ability of the safety filter to preserve the performance of a nominal controller while limiting the current magnitude. Lastly, Section \ref{section:conclusion} concludes the paper.



\section{System Model and Control Objective}
\label{section:model}
\subsection{Inverter Model} 
In this paper, we consider a three-phase inverter connected to an infinite bus via an \emph{RL} branch. The model assumes we have balanced three phase and can use a direct-quadrature (dq) reference frame with reference to 
some rotating angle to describe all the rotating physical quantities~\cite{yazdani_vsc_book}. This model is commonly used when studying the transient stability of inverters controlled as a voltage source.

\begin{figure}[ht]
    \centering
    \includegraphics[width=0.9\columnwidth, trim={0 0 0 0},clip]{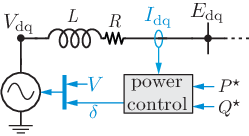}
    \caption{Simplified inverter system model under study.}
    \label{fig:simplified_model}
    \vspace*{-5pt}  
\end{figure}

The dynamical system model of a voltage-source inverter connected to a stiff grid through an RL filter is
\begin{equation} \label{eqn:nonlinear-system}
\begin{bmatrix}\dot{I_{\mathrm{d}}}\\\dot{I_{\mathrm{q}}}\end{bmatrix}=  \begin{bmatrix}- \frac{R}{L} & \omega\\- \omega & - \frac{R}{L}\end{bmatrix} 
\begin{bmatrix}I_{\mathrm{d}}\\I_{\mathrm{q}}\end{bmatrix}
+ \frac{1}{L} \left(\begin{bmatrix} 
V \cos(\delta) \\
V \sin(\delta)
\end{bmatrix} 
- E_\mathrm{dq} \right),
\end{equation}
where 
$R$ and $L$ are the resistance and impedance of the \emph{RL} filter, respectively, $E_\mathrm{dq}$ is the grid-side voltage with reference to some angle $\theta$, $\omega$ is the rotational frequency of the reference frame, $V$ is the inverter voltage magnitude, and $\delta$ is the angle difference between the inverter voltage and $\theta$. 
Making a small-angle assumption, such that $\cos(\delta) \approx 1$ and $\sin(\delta) \approx \delta$, assuming the reference frame rotates near to a nominal value $\omega_\mathrm{nom}$, and assuming $E_\mathrm{dq} = (E, 0)$, where $E$ is the grid-side voltage magnitude, this reduces to the linear system \cite{joswigjones2024optimal, Pogaku_Prodanovic_Green_2007, Zhou_Liu_Zhou_She_2015} 
\begin{equation} \label{eqn:linear-system}
\dot{x}= A x + B u, 
\end{equation}
where 
\begin{align*}
     A = \begin{bmatrix}- \frac{R}{L} & \omega_\mathrm{nom}\\ -\omega_\mathrm{nom} & - \frac{R}{L}\end{bmatrix}, B = \begin{bmatrix}0 \\ \frac{V}{L}\end{bmatrix},
     x = \begin{bmatrix}I_{\mathrm{d}}\\I_{\mathrm{q}}\end{bmatrix}, \;\; u = \delta.
\end{align*} 
For this work, we set $V = E$ and take $\delta$ to be the only control input to the system. The active and reactive power outputs from the inverter are given by 
\begin{align*}
P &= \frac{3}{2} \left( V \cos{(\delta)} I_\mathrm{d} + V \sin{(\delta)} I_\mathrm{q}  \right), \\ 
Q &= \frac{3}{2} \left( V \sin{(\delta)} I_\mathrm{d} - V \cos{(\delta)} I_\mathrm{q} \right).
\end{align*}
With the the above notation and assumptions 
, tracking a given $P^*,Q^*$ is equivalent to tracking some $I_\mathrm{d}^*$ and $I_\mathrm{q}^*$.

\subsection{Safety and Stability} 

To protect internal switching devices, the inverter's output current cannot exceed a given limit~\cite{Fan_Liu_Zhao_Wu_Wang_2022}.
This leads to a bound on the magnitude of the inverter's output current in the dq reference frame, $\left| I 
\right|_\mathrm{max}$. To satisfy this bound we must always have $\|x\|_2 \leq \left| I 
\right|_\mathrm{max}$. We define a set of safe states in terms of the currents magnitude bounds as $\mathcal{S} \coloneqq \{x \hquad|\hquad h(x) \geq 0\}$, where 
$$
h(x) = \left| I 
\right|_\mathrm{max}^2 - \|x\|_2^2,
$$ 
and the set of states at the boundary of this set as $\partial \mathcal{S} \coloneqq \{x \hquad|\hquad h(x) = 0\}$. We want the inverter to be able to operate near or at this bound to maximize the use of the inverter's capabilities especially during faults or transients when the grid requires support. 

A control action is safe if it pushes the states such that they do not leave the safe set anytime they are at the boundary of the safe set.  Mathematically, the set of safe control inputs for some $x \in \mathcal{S}$, with dynamics $\dot{x} = f(x) + g(x) \cdot u$, can be defined as the values of $u$ that satisfy the control barrier function (CBF) constraint 
\begin{equation}
\label{eq:cbf-inequality}
\dot{h}(x) = \nabla h(x)^\top (f(x) + g(x) \cdot u) \geq -\alpha (h(x)) \hquad\forall\hquad x \in \mathcal{S},
\end{equation} where $\nabla h(x)$ is the gradient of $h(x)$ with respect to $x$, and $\alpha$ is a class $\mathcal{K}$ function (strictly increasing and $\alpha(0) = 0$)~\cite{Ames_Xu_Grizzle_Tabuada_2017}. The inclusion of the term $\alpha(h(x))$ allows the safety condition to be applied to the entire safe region instead of just at its boundary. We note that for our system $f(x) = A x$ and $g(x) = B$.

To control the power output of the inverter to a desired value, we need the system to be able to track to a given set point, $x^*$. Formally, we require the control to be asymptotically stable with respect to some control Lyapunov function (CLF) \cite{doi:10.1137/0321028, ARTSTEIN19831163}. Given a CLF, $V(x)$, the set of stabilizing control inputs for some $x \in \mathcal{S}$ can be defined as the values of $u$ satisfying the constraint 
\begin{equation}
\label{eq:clf-inequality}
\dot{V}(x) = \nabla V(x)^\top (f(x) + g(x) \cdot u) \leq 0.
\end{equation}

We want to control our system such that $\mathcal{S}$ is an invariant set and the control is stabilizing (i.e. given an initial condition $x_0 \in \mathcal{S}$ the system will converge to a feasible $x^*$ without the states leaving $\mathcal{S}$ at any point during their trajectory). Typically, a nominal voltage controller $u_\mathrm{nom}(x)$ that is stabilizing will not satisfy the safety constraint for references near the current magnitude boundary and will produce unsafe control actions.





\section{Current Limiting Safety Filter}
\label{section:control}
We assume that the reference point is feasible: $\|x^*\|_2 \leq \left| I 
\right|_\mathrm{max}$, and $u^*, x^*$ satisfy 
\begin{equation}
\label{eq:feasible-reference}
A x^* + B u^* = 0. 
\end{equation}
Our goal is to design a feedback controller such that the system in \eqref{eqn:linear-system} is stable with respect to a given setpoint $x^*$ ($x(t) \rightarrow x^*$ as $t \rightarrow \infty$), and is safe with respect to the current magnitude limit. 

\subsection{Safety Filter Formulation}
A nonlinear control method often used to achieve safety guarantees with a well performing controller is a \textit{safety filter}. A safety filter is a function that is applied to some control action that can detect unsafe control inputs that may lead to constraint violations and minimally modifies them to ensure safety~\cite{Wabersich_Taylor_Choi_Sreenath_Tomlin_Ames_Zeilinger_2023}. In this section we introduce how such a safety filter can be applied to our system and provide guarantees on the feasibility of the safety filter problem constraints.

We can formulate this safety filter as a quadratic program (QP) as follows \cite{Ames_Grizzle_Tabuada_2014}.
\begin{subequations} 
\label{eq:safety-filter}
\begin{align}
\label{eq:safety-filter-problem}
\displaystyle
\bar{u} = \argmin_{u} & \| u - u_\mathrm{nom}(x) \|_2^2  \tag{QP} \\
\textrm{s.t.} \quad &
\nabla h(x)^\top ( f(x) + g(x) \cdot u ) \geq -\alpha h(x) \tag{CBF} \label{eq:cbf-constraint}\\
\quad &
\nabla V(x)^\top ( f(x) + g(x) \cdot u ) \leq 0 \tag{CLF}, \label{eq:clf-constraint}
\end{align}
\end{subequations}
where $\bar{u}$ is the filtered voltage control action, and we select $\alpha$ to be a positive constant. For this paper we use the CLF $V(x) = (x - x^*)^\top (x - x^*)$. 

Although QPs are considered to be simple optimization problems, using iterative solvers in real-time on inverters themselves is not trivial, because of microprocessor limitations and sampling/input delays~\cite{Singletary_Chen_Ames_2020, Buso_Mattavelli_2015}. However, this form of QP has a closed-form solution~\cite{Ames_Xu_Grizzle_Tabuada_2017}. In our case, because the optimization is over a scalar variable $u$, there is actually a simple closed-form solution (see Algorithm~\ref{alg:closed-form-qp}). This algorithm comes from projecting a point into an interval, and can be implemented on existing inverter microcontrollers. 
\begin{algorithm}
\caption{Closed-form solution to \ref{eq:safety-filter-problem} with $u \in \mathbb{R}^1$}\label{alg:closed-form-qp}
\begin{algorithmic}
\Require $x_t$, $u_t$, $a_t^\mathrm{cbf}$, $b_t^\mathrm{cbf}$, $a_t^\mathrm{clf}$, $b_t^\mathrm{clf}$
\Ensure $\bar{u}_t$
\State $a_t^\mathrm{cbf} \gets \nabla h(x_t)^\top g(x_t)$
\State $b_t^\mathrm{cbf} \gets  -\alpha h(x_t) - \nabla h(x_t)^\top f(x_t)$
\State $a_t^\mathrm{clf} \gets \nabla V(x_t)^\top g(x_t)$
\State $b_t^\mathrm{clf} \gets - \nabla V(x_t)^\top f(x_t)$
\State $u_{lb} \gets -\infty$
\State $u_{ub} \gets \infty$
\If{$a_t^\mathrm{cbf} u_t \geq b_t^\mathrm{cbf}$ and $a_t^\mathrm{clf} u_t \leq b_t^\mathrm{clf}$}
    \State $\bar{u}_t \gets u_t$
\Else{}

    \If{$a_t^\mathrm{cbf} \geq 0$}
        \State ${u_{lb}} \gets \max{(u_{lb}, b_t^\mathrm{cbf} / a_t^\mathrm{cbf})}$
    \Else{}
        \State ${u_{ub}} \gets \min{(u_{ub}, b_t^\mathrm{cbf} / a_t^\mathrm{cbf})}$
    \EndIf
    \If{$a_t^\mathrm{clf} \geq 0$}
        \State ${u_{ub}} \gets \min{(u_{ub}, b_t^\mathrm{clf} / a_t^\mathrm{clf})}$
    \Else{}
        \State ${u_{lb}} \gets \max{(u_{lb}, b_t^\mathrm{clf} / a_t^\mathrm{clf})}$
    \EndIf

    \State $\bar{u}_t \gets \min{(u_{ub}, \max{(u_t, u_{lb})})}$
\EndIf
\end{algorithmic}
\end{algorithm}
\vspace*{-10pt}  

When using an optimization-based controller it is important to ensure that there exists at least one feasible solution at any possible state of the system. 
This is especially relevant to our problem where the safety and stability constraints often can seem to be working towards opposite goals 
(we want a controller that can track current references near the current magnitude limit, but never exceeds this current magnitude limit).
In general, there is no guarantee that there is a feasible input to this safety filter problem for a generic linear system, but given the properties of our system and safety constraint we can guarantee the existence of a feasible solution.
In the following section we prove that there always exists a feasible control action that satisfies (\ref{eq:cbf-constraint}) and (\ref{eq:clf-constraint}) for (\ref{eqn:linear-system}).


\subsection{Safety Filter Feasibility}
The theorem below states the main result of the paper.
\begin{theorem} \label{thm:main}
Given the linear system in \eqref{eqn:linear-system} and assume that the setpoints $x^*$ and $u^*$ are feasible. Consider the magnitude barrier function $h(x) = \left| I 
\right|_\mathrm{max}^2 - \|x\|_2^2$. 
If $A + A^\top \prec 0$, and $A^{-1} B \ne 0$,
then a safe and stable control actions always exists such that if $h(x_0) \geq 0$ then $h(x_t) \geq 0$ and $x_t \to x^*$.
\end{theorem}


Our approach to proving that there always exists a feasible solution to the safety filter problem is to prove the existence of a safe and stable linear feedback controller, $K$, for a given $x^*, u^*$. Because this controller exists, there is at least one feasible action $u$ that satisfies (\ref{eq:cbf-constraint}) and (\ref{eq:clf-constraint}), namely, the linear control action, $u = u^* - K (x - x^*)$. The main condition on the system is that the $A$ matrix is "stable" enough, namely $A+A^T \prec 0$ (this can be seen as the Lyapunov stability condition for linear systems with a Lyapunov function of $||x||_2^2$). This result is intuitive in the following sense. Because the level sets of the Lyapunov function and the barrier function have the same shape---they are both circular---the tension between stability and safety can be resolved. The formal proof of this theorem is nontrivial, and the details are given in Section~\ref{sec:proof}.

Even though a safe and stable linear controller exists, there may be more optimal control actions that are safe and stabilizing. 
Imposing a safe linear controller over the entire safe region is overly conservative and can result in a suboptimal performance.
Therefore, we are motivated to look at the nonlinear safety filter controller when the control objective includes safety constraints and performance specifications. 
The ability of this safety filter to provide more optimal safe control actions compared to a safe linear controller can be seen in Section~\ref{section:results}.



\section{Proof of Theorem~\ref{thm:main}}
\label{sec:proof}

We have two major lemmas that establishes the existence of a safe and stable $K$. 

\begin{lemma}[Conditions for a safe \& stable system]
\label{lemma:safety}
Given $M \in \mathbb{R}^{p \times p}$, $y \in \mathbb{R}^{p}$, and $z \in \mathbb{R}^{p}$,
where $\|y\|_2 \leq \|z\|_2$,
if $M^\top y = \lambda y$, $\lambda_\mathrm{min}(M + M^\top) \geq \lambda$, and $M + M^\top \succeq 0$,
then
$$
z^\top M z - z^\top M y \geq 0.
$$
\end{lemma}
\begin{lemma}[Existence of a safe  \& stable $K$]
\label{lemma:existance}
Given $A \in \mathbb{R}^{p \times p}, B \in \mathbb{R}^{p \times 1}$. Given $x^*$ and $u^*$, satisfying $A x^* + B u^* = 0$, 
if $A + A^\top \prec 0$, and $A^{-1} B \ne 0$, 
then there exists a $K \in \mathbb{R}^{1 \times p}$, such that the closed-loop system $N \coloneq A - B K$ has the properties
$N^\top x^* = \lambda x^*$, $\lambda_\mathrm{max}(N + N^\top) \leq \lambda$, and $N + N^\top \preceq 0$.  
\end{lemma}

We can apply these lemmas to prove \textbf{Theorem}~\ref{thm:main} as follows. Since we have $x^*, u^*$ satisfying (\ref{eq:feasible-reference}), then the closed-loop dynamics of $(\ref{eqn:linear-system})$ with control $u = u^* - K (x - x^*)$ can be equivalently represented as $\dot{x} = -\bar{A} (x - x^*)$, where $\bar{A} = -(A - B K)$. For safety of this equivalent system we require (\ref{eq:cbf-inequality}) to hold for all $x \in \partial \mathcal{S}$, with $f(x) = -\bar{A} (x - x^*), g(x) = 0$. 
This is equivalent to the inequality 
$$x^\top \bar{A} x - x^\top \bar{A} x^* \geq 0 \hquad\forall\hquad x \in \partial\mathcal{S}.
$$
We can show this inequality holds using \textbf{Lemma}~\ref{lemma:safety}, where $M = \bar{A}$, $z = x$ and $y = x^*$, if we can find a $K$ that satisfies the following criteria:
\begin{subequations}
 \label{eq:safety-criteria}
\begin{align} 
    (x^*)^\top (A - B K) &= (x^*)^\top \lambda, \label{eq:safety-criteria-1} \\ 
    \lambda_\mathrm{max}((A - B K) + (A - B K)^\top) &\leq \lambda, \label{eq:safety-criteria-2} \\
    (A - B K) + (A - B K)^\top &\prec 0, \label{eq:safety-criteria-3}
\end{align}
\end{subequations}
where $\lambda \in \mathbb{R}$ is the eigenvalue associated with $x^*$. With \textbf{Lemma}~\ref{lemma:existance}, we can see that for our system a $K$ satisfying (\ref{eq:safety-criteria}) always exists, as the $A$ matrix of (\ref{eqn:linear-system}) meets the SDP condition, $A + A^\top \prec 0$, and the system is controllable. 

Once these two lemmas are established we can conclude that there exists a safe linear feedback controller such that there always exists at least one feasible solution for the safety filter (\ref{eq:safety-filter-problem}) with our system (\ref{eqn:linear-system}).

\subsection{Proof of Lemma 1}

\begin{proof}[Proof of Lemma \ref{lemma:safety}]
We start by noting that $z^\top M z - z^\top M y \geq 0$ is equivalent to $(z - y)^\top M (z - y) \geq y^\top M y - y^\top M z $.
%
%
Since $y$ is a left eigenvector of $M$ with eigenvalue $\lambda$ this inequality is equivalent to
$$
(z - y)^\top M (z - y) \geq \lambda y^\top y - \lambda y^\top z.
$$
Noting that $(z - y)^\top M (z - y)$ = $(z - y)^\top \frac{M + M^\top}{2} (z - y)$, we have
$$
(z - y)^\top \frac{M + M^\top}{2} (z - y) \geq \lambda y^\top y - \lambda y^\top z. 
$$
We define $v_1, v_2$ and $\lambda_1, \lambda_2$ to be the eigenvectors and eigenvalues of $\frac{M + M^\top}{2}$, and note that we can pick $v_1, v_2$ to form an orthonormal basis as $\frac{M + M^\top}{2}$ is real and symmetric. We define $c_1 = (z - y)^\top v_1$, $c_2 = (z - y)^\top v_2$ such that $(z - y) = c_1 v_1 + c_2 v_2$ and our inequality is equal to
$$
\lambda_1 c_1^2 + \lambda_2 c_2^2 \geq \lambda y^\top y - \lambda y^\top z 
$$
Rearranging $(z - y) = c_1 v_1 + c_2 v_2$ we have $z = y + c_1 v_1 + c_2 v_2$ which we substitute in for the remaining $z$ to get
$$
\lambda_1 c_1^2 + \lambda_2 c_2^2 \geq - \lambda (y^\top (c_1 v_1 + c_2 v_2)).
$$
Using our assumption that $\lambda_1, \lambda_2 \geq \frac{\lambda}{2}$ we can see that the left-hand side of the inequality is lower bounded as
$$
\lambda_1 c_1^2 + \lambda_2 c_2^2 \geq \frac{\lambda}{2} (c_1^2 + c_2^2)  \geq - \lambda (y^\top (c_1 v_1 + c_2 v_2)),
$$
which is equivalent to showing that 
$$
c_1^2 + c_2^2 \geq -2 (c_1 y^\top v_1 + c_2 y^\top v_2),
$$ as $\lambda_1 c_1^2, \lambda_2 c_2^2 \geq 0$ ($M + M^\top \succ 0 \implies \lambda_1, \lambda_2 > 0$). This can be shown using \textbf{Lemma}~\ref{lemma:ortho_basis_ineq}, with $\zeta = z, \gamma = y$, and $\{w_1, w_2\} = \{v_1, v_2\}$, as by assumption $\|z\|_2 \geq \|y\|_2$.
\end{proof}

\begin{lemma}
\label{lemma:ortho_basis_ineq}
Given an orthonormal basis $\{w_1, w_2\}$, and $\gamma, \zeta \in \mathbb{R}^{2}$ such that $\|\zeta\|_2 \geq \|\gamma\|_2$,
$$
c_1^2 + c_2^2 \geq -2 (c_1 b + c_2 d),
$$
where $c_1 = a - b$, $c_2 = c - d$, and $a = \zeta^\top w_1, b = \gamma^\top w_1, c = \zeta^\top w_2, d = \gamma^\top w_2$.
\end{lemma}
\begin{proof}
Plugging the definitions for $c_1$ and $c_2$ into the inequality and expanding we have
$$(a^2 - 2ab + b^2) + (c^2 - 2 c d + d^2) \geq -2 (a b - b^2 + c d - d^2).$$
Noting that $a^2 + c^2 = \|\zeta\|_2^2$ and $b^2 + d^2 = \|\gamma\|_2^2$, as $w_1 \perp w_2$, this reduces to 
$$\|\zeta\|_2^2 + \|\gamma\|_2^2 \geq  2 \|\gamma\|_2^2,$$
which is true as we assume $\|\zeta\|_2 \geq \|\gamma\|_2$. 
\end{proof}

\subsection{Proof of Lemma 2}
Because of space constraints, we only show this lemma when the control input is scalar (i.e., $u \in \mathbb{R}$). The proof for the vector case is considerably more involved and not needed for the results in this paper. 
\begin{proof}[Proof of Lemma~\ref{lemma:existance}]
First we note that $x^* = A^{-1} B u^*$ such that $N^\top x^* = \lambda x^* \iff u^* B^\top A^{-\top} (A - B K) = \lambda u^* B^\top A^{-\top}$. Assuming $u^* \ne 0$ and dividing it out we get 
$$
B^\top A^{-\top} (A - B K) = \lambda B^\top A^{-\top}.
$$
Solving for $K$, noting that $B^\top A^{-\top} B$ is a scalar, we get
$$
K = \frac{1}{B^\top A^{-\top} B} \left( B^\top A^{-\top} (A - \lambda I) \right).
$$
Plugging this into our expression for $N$ we have
$$
A - B K = A - \frac{1}{B^\top A^{-\top} B} \left( B B^\top A^{-\top} (A - \lambda I) \right),
$$
which can be substituted into the inequality $A - B K + (A - B K)^\top \prec \lambda I$ to get
$$
M_1 - \lambda M_2 \prec 0,
$$
where 
\begin{align*}
M_1 =& A + A^\top \\
&- \frac{1}{B^\top A^{-\top} B} \left( B B^\top A^{-\top} A + A^\top A^{-1} B B^\top \right), \\
M_2 =& I - \frac{1}{B^\top A^{-\top} B} \left( B B^\top A^{-\top} + A^{-1} B B^\top \right).
\end{align*}
We will show the following such that we can always choose $\lambda < 0$, sufficiently small, to have $M_1 - \lambda M_2 \prec 0$: $M_1 \preceq 0$, with a single zero eigenvalue associated with the eigenvector $z$ (i.e. $M_1 z = 0$), and $z^\top M_2 z < 0$.

First, we show that $M_1 \preceq 0$, using the fact that if $X \prec 0$, then $P^\top X P \preceq 0 \hquad\forall\hquad P$. To construct this form consider $P = I - \frac{A^{-1} B B^\top A^{-\top} A}{B^\top A^{-1} B}$. By direct computation, $M_1 = P^\top (A + A^\top) P$.  This implies that $M_1 \preceq 0$ as $(A + A^\top) \prec 0$.

Now, we can show that $M_1$ has only one zero eigenvalue using this same $P^\top X P$ form. If $v$ is a zero eigenvector of $M_1$ then we have that $v^\top P^\top (A + A^\top) P v = 0$. Since we assume $A + A^\top \prec 0$, then this implies that $P v = 0$. Note that $P$ is the sum of a rank $n$ matrix, $I$, and a rank 1 matrix, which implies that $P$ is at least rank $n - 1$ and can have at most one zero eigenvalue. We see that this zero eigenvalue is associated with the eigenvector $z = A^{-1} B$ by direct computation; $M_1 z = 0$ using the equality $B^\top A^{-1} B = B^\top A^{-\top} B$, which holds as $B^\top A^{-1} B$ is a scalar. 

Lastly, we can show that $z^\top M_2 z < 0$ as
\begin{align*}
z^\top M_2 z =& (A^{-1} B)^\top M_2 (A^{-1} B) \\
=& B^\top A^{-\top} A^{-1} B \\
&- \frac{1}{B^\top A^{-1} B} \big( (B^\top A^{-\top} B) B^\top A^{-\top} A^{-1} B \\
& \hspace{5.45em} + B^\top A^{-\top} A^{-1} B (B^\top A^{-1} B) \big) \\
\stackrel{(a)}{=}& -B^\top A^{-\top} A^{-1} B \\
=& - \| A^{-1} B \|_2^2 < 0
\end{align*}
where $(a)$ follows from $B^\top A^{-1} B = B^\top A^{-\top} B$.
\end{proof}


\section{Simulation Results}
\label{section:results}
To test this approach we simulate the response of the inverter system for a range of initial conditions and power references. We test the system with an LQR controller, a safe linear $K$ controller, and the LQR controller with CBF safety filter. The parameters used for the system in these simulations can be found in Table~\ref{table:parameters}. The code used to generate these results is available at \url{https://github.com/TragerJoswig-Jones/Safe-Current-Magnitude-Limit-Inverter-Control}.

\begin{table}
\caption{Inverter \& \emph{RL} filter system parameters.}
\begin{center}
\begin{tabular}{c|c}
Parameter & Value \\
\hline
$V_\mathrm{nom}$    & 120~\unit{\volt} \\ 
\hline
$S_\mathrm{nom}$    & 1.5~\unit{\kilo{\volt\ampere}} \\
\hline
$I_\mathrm{nom}$    & 4.17~\unit{\ampere} \\
\hline
$\left| I 
\right|_\mathrm{max}$    & 5~\unit{\ampere} \\
\hline
\end{tabular}
\hspace{1em}
\begin{tabular}{c|c}
Param. & Value \\
\hline
$E$   & 120~\unit{\volt} \\
\hline
$\omega_\mathrm{nom}$    & $2\pi60$~\unit{\radian\per\s} \\
\hline
$R$         & 1.3~\unit{\ohm} \\
\hline
$L$     & 3.5~\unit{\milli\henry} \\
\hline
\end{tabular}
\label{table:parameters}
\end{center}
\vspace*{-10pt}  
\end{table}

The nominal LQR controller is computed with $Q = I$, $R = V_\mathrm{nom} / (10 \cdot L)$ and is found to be
$
K_\mathrm{LQR} = \begin{bmatrix}
0.0009 & 0.0099
\end{bmatrix}.
$
To design a safe linear feedback controller we use the conditions of \textbf{Lemma}~\ref{lemma:safety} to formulate the optimization problem
\begin{equation}
\begin{aligned}
\label{eq:safe-linear-ctrl-problem}
\displaystyle
\min_{K, \lambda} &\quad \| K \|_2 \\
\textrm{s.t.} &\quad
(x^*)^\top (A - B K) = (x^*)^\top \lambda, \\
&\quad \lambda_\mathrm{max}((A - B K) + (A - B K)^\top) \leq \lambda, \\
&\quad (A - B K) + (A - B K)^\top \prec 0, \\
\end{aligned}
\end{equation}
where we choose to minimize the spectral norm of $K$ as the objective. Note that it is difficult to select an objective that designs a controller for performance with this form of static controller optimization problem. We also note that these constraints are convex in $K$ and $\lambda$. We solve (\ref{eq:safe-linear-ctrl-problem}) using \texttt{CVXPY} \cite{diamond2016cvxpy, agrawal2018rewriting} to get
$
K = \begin{bmatrix}
-0.0111 &  0.0111
\end{bmatrix}.
$
The CBF safety filter, (\ref{eq:safety-filter-problem}), is ran with the nominal control action, $u_\mathrm{nom} = u^* - K_\mathrm{LQR} (x - x^*)$, and $\alpha = 1000$, and is solved using Algorithm~\ref{alg:closed-form-qp}. Simulations are performed using the odeint solver from \texttt{SciPy} with a sampling time of $\Delta{t} = 10~\unit{\micro\second}$ and simulation length of $t_\mathrm{end} = 50~\unit{\milli\second}$. The costs for these tests are calculated as 
$$
\mathrm{cost}_i = 1000\cdot \sum_{t=0}^{t_\mathrm{end}} \Delta{t} \left( \tilde{x}_{i,t}^\top Q \tilde{x}_{i,t} + \tilde{u}_{i,t}^\top R \tilde{u}_{i,t} \right),
$$
where $\tilde{x}_{i,t}= x_{i,t} - x^*_{i,t}$ and $\tilde{u}_{i,t} = u_{i,t} - u^*_{i,t}$ are the state and input errors at time $t$ for test $i$. 

In the following we use the function $\mathrm{lin}(a, b, n) \coloneqq a + \frac{{n-1}}{{i-1}}(b-a)$ to define a list of evenly spaced numbers over a specified interval. 


\subsection{\texorpdfstring{$x_0 \in \partial\mathcal{S}$}{Initial conditions at boundary} tests}
\label{sec:x0indS}
We test the three controllers for a range of initial conditions from $\partial \mathcal{S}$, specifically 
$X_0 \coloneqq \{(\left| I 
\right|_\mathrm{max}\sin(\phi), \left| I 
\right|_\mathrm{max}\cos(\phi])) \hquad|\hquad \phi \in \mathrm{lin}(0, 2 \pi - \frac{2 \pi}{100}, 100) \}.$
We select a single $x^*$ reference value for these tests, that is $x^* = (3.56~\unit{\ampere},3.51~\unit{\ampere})$, by selecting a feasible $x^*$ value from (\ref{eq:feasible-reference}) with magnitude $|I|_\mathrm{max}$. 

The average cost for each controller can be found in Table~\ref{table:traj-costs-dS}.
The nominal LQR controller always performs optimally, but also always has a trajectory that is unsafe in these tests. The safe $K$ controller is always safe, but performs poorly for the given LQR costs. Lastly, the CBF safety filter controller has a similar, yet slightly worse, performance to the LQR controller, but the trajectory always remains safe due to the safety filter. A selected example trajectory can be seen in Fig~\ref{fig:sim_state_trajectory}. Notably, the CBF trajectory aligns with the LQR trajectory up until the safety constraint would be violated and the safety filter intervenes. The alteration of the control action that causes this divergence in trajectories can be seen in Fig~\ref{fig:sim_input_trajectory}.

\begin{table}
\caption{Average cost with $x^* = (3.56~\unit{\ampere},3.51~\unit{\ampere})$ and $x_0 \in X_0$.}
\label{table:traj-costs-dS}
\begin{center}
\begin{tabular}{r|c|c|c}
\textbf{controller} & $K$ & CBF & LQR \\
\hline 
\textbf{cost} & 82.22 & 59.16 & 58.57  \\
\hline
\end{tabular}
\end{center}
\vspace*{-12pt}  
\end{table}

\begin{figure}
    \centering
    \includegraphics[width=\columnwidth, trim={0 10 0 0},clip]{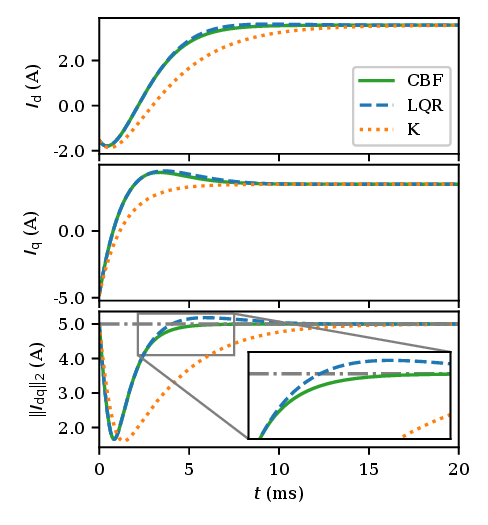}
    \caption{$I_\mathrm{dq}$ trajectories for CBF, LQR, and $K_\mathrm{safe}$ controllers with $x_0 = (-1.55~\unit{\ampere}, -4.76~\unit{\ampere}), x^* = (3.56~\unit{\ampere},3.51~\unit{\ampere})$. The LQR control is seen to exceed the safety bound, the safe $K$ control has costly performance, and the CBF control performs well and remains safe.}
    \label{fig:sim_state_trajectory}
    \vspace*{-5pt}  
\end{figure}

\begin{figure}[ht]
    \centering
    \includegraphics[width=\columnwidth, trim={0 10 0 0},clip]{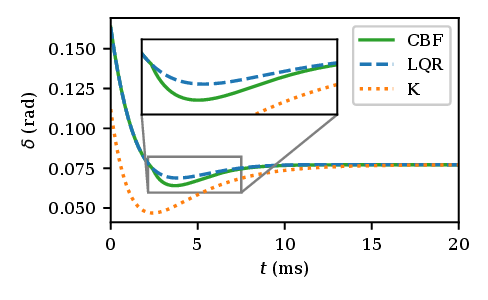}
    \caption{$\delta$ input values for CBF, LQR, and $K_\mathrm{safe}$ controllers with $x_0 = (-1.55~\unit{\ampere}, -4.76~\unit{\ampere}), x^* = (3.56~\unit{\ampere},3.51~\unit{\ampere})$.}
    \label{fig:sim_input_trajectory}
    \vspace*{-5pt}  
\end{figure}

\subsection{Sampled \texorpdfstring{$x_0, x^*$}{initial conditions and reference} tests}
We then test these controllers with randomly sampled $x_0$ and $x^*$ values. We sample $x^*$ reference values by selecting a feasible $x^*$ value from (\ref{eq:feasible-reference}) and scaling the magnitude to a uniformly sampled value from the range $[-\left| I 
\right|_\mathrm{max}, \left| I 
\right|_\mathrm{max}]$. $x_0$ values are calculated as $x_0 = (r_0 \cos{(\phi_0)}, r_0 \sin{(\phi_0)})$ with $r_0, \phi_0$ uniformly sampled from the ranges $[0, \left| I 
\right|_\mathrm{max}]$ and $[0, 2 \pi]$, respectively.
%
%
\begin{figure}[ht]
    \centering
    \includegraphics[width=\columnwidth, trim={0 8 0 0},clip]{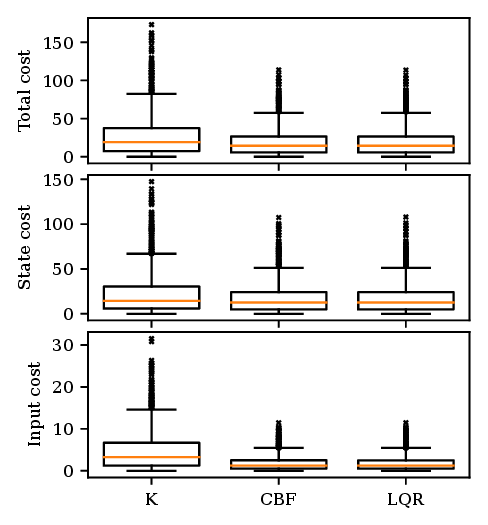}
    \caption{Costs of CBF, LQR, and $K_\mathrm{safe}$ controllers for 1,000 randomly sampled $x_0$ and $x^*$ values.}
    \vspace*{-5pt}  
    \label{fig:sampled_costs}
\end{figure}

Testing with 1,000 randomly sampled $x_0$ and $x^*$ values we see a similar trend in the average costs
with $K$ performing the worst, LQR control performing the best, and the CBF control performing slightly worse than the LQR controller. The distribution of costs from these tests can be seen in Fig~\ref{fig:sampled_costs}. Of these tests 24 of them had unsafe trajectories with LQR control, while all tests were safe with the other two controllers. The difference in the LQR and CBF controller costs are small, but for all tests the CBF controller cost is equal or slightly higher than the LQR controller cost due to the safety filter intervening.

\subsection{Small-angle assumption}
Lastly, we test the validity of the small-angle assumption with the CBF control. In this section we compare the linear system used for analysis to the nonlinear system without the small-angle assumption (\ref{eqn:nonlinear-system}). For the nonlinear system we use the CBF control formulated for the linear system. We subject the two systems to the same range of initial conditions as in Section~\ref{sec:x0indS}, but take $x^*$ to be (3.42~\unit{\ampere},3.64~\unit{\ampere}), a feasible reference value for the nonlinear system (\ref{eqn:nonlinear-system}).

An exemplary case can be seen in Fig~\ref{fig:nonlinear_sim_state_trajectory}. The trajectories of the two systems do not differ significantly for any of the test cases, demonstrating that the small-angle assumption holds well with the CBF control. However, the trajectories do differ and the nonlinear system slightly exceeds the current magnitude bound by up to 0.5\% in some of the tests. Further the nonlinear system is seen to not fully converge to the reference value due to the safety filter. Fig~\ref{fig:nonlinear_sim_input_trajectory} shows the input trajectories are nearly aligned, but when the safety filter activates at 2~\unit{\milli\second} it introduces an offset between the inputs which prevents the states of the nonlinear system from completely approaching the magnitude bound.

These issues seen when applying the linear CBF control to the nonlinear system are small, but may not be tolerable. The boundary violation issue can be mitigated by setting the boundaries to be slightly tighter than the actual current limit values. A solution to resolve both issues would be to formulate the CBF control for the nonlinear system. However, this would introduce nonlinearity into the constraints of the safety filter problem making it more computationally expensive to solve in real time. 

\begin{figure}[ht]
    \centering
    \includegraphics[width=\columnwidth, trim={0 10 0 0},clip]{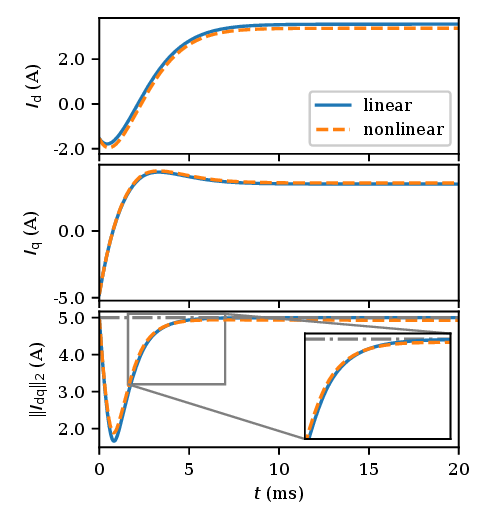}
    \caption{$I_\mathrm{dq}$ trajectories for for the linear and nonlinear systems with CBF control and $x_0 = (-1.55~\unit{\ampere}, -4.76~\unit{\ampere}), x^* = (3.42~\unit{\ampere},3.64~\unit{\ampere})$. The linear and nonlinear systems have a similar trajectories.}
    \label{fig:nonlinear_sim_state_trajectory}
    \vspace*{-5pt}  
\end{figure}

\begin{figure}[ht]
    \centering
    \includegraphics[width=\columnwidth, trim={0 10 0 0},clip]{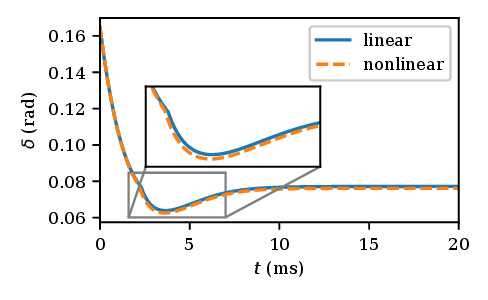}
    \caption{$\delta$ input values for the linear and nonlinear systems with CBF control and $x_0 = (-1.55~\unit{\ampere}, -4.76~\unit{\ampere}), x^* = (3.42~\unit{\ampere},3.64~\unit{\ampere})$.}
    \label{fig:nonlinear_sim_input_trajectory}
    \vspace*{-10pt}  
\end{figure}

\section{Conclusions}
\label{section:conclusion}
In this work we demonstrate a method for current limiting of inverter-based resources controlled as a voltage sources that can safely limit the current while minimally altering the nominal control action. 
This safety filter approach is proven to always be feasible with the dynamics of a \emph{RL} branch and can nearly maintain the performance of an unsafe nominal controller. 
Future work includes considering the impact of disturbances and parameter uncertainty, extending these results to include voltage magnitude control, modifying this method for integral control, and studying the stability of networks of inverters using this approach. 




\section*{Acknowledgement}
The authors were partially supported by NSF grant ECCS-2023531 and the State of Washington through the University of Washington Clean Energy Institute.






\printbibliography

\end{document}